\documentclass[12pt]{article}
\usepackage{graphicx}
\usepackage{amssymb}
\usepackage[latin1]{inputenc}

\newcommand{\newc}{\newcommand}
\newc{\beq}{\begin{equation}}
\newc{\eeq}{\end{equation}}
\newc{\bea}{\begin{array}}
\newc{\eea}{\end{array}}
\newcommand{\ben}{\begin{eqnarray}}
\newcommand{\een}{\end{eqnarray}}
\newc{\ra}{\rightarrow}
\newc{\bfx}{{\bf x}}
\newc{\bfV}{{\bf V}}
\newc{\cO}{{\cal O}}
\newc{\bfv}{{\bf v}}
\newc{\bfu}{{\bf u}}
\newc{\bfp}{{\bf p}}
\newc{\ve}{{\varepsilon}}
\newc{\Psibar}{\overline\Psi}
\newc{\w}{{\bf w}}
\newc{\E}{{\mathbf{E}}}
\newc{\EE}{{\mathcal E}}
\newc{\bfn}{{\mathbf\nabla}}
\newc{\la}{{\cal L}}
\newc{\tla}{{\tilde{\cal L}}}
\newc{\bp}{{\bf p}}
\newc{\ho}{\hookrightarrow }
\newc{\bP}{{\bf P}}
\newc{\pd}{{\partial}}
\newc{\piv}{{\partial_4}}
\newc{\pv}{{\partial_5}}
\newc{\bJ}{{\bf J}}
\newc{\bze}{{\mathbf 0}}
\newc{\bK}{{\bf K}}
\newc{\tphi}{{\tilde\phi}}
\newc{\tF}{{\tilde F}}
\newc{\tD}{{\tilde D}}
\newc{\tJ}{{\tilde J}}
\newc{\tj}{{\tilde j}}
\newc{\bD}{{\bf D}}
\newc{\tvphi}{{\tilde\varphi}}
\newc{\trho}{{\tilde\rho}}
\newc{\ttheta}{{\tilde\theta}}
\newc{\tpsi}{{\tilde\psi}}
\newc{\tu}{{\tilde u}}
\newc{\cD}{{\cal D}}
\newc{\tPhi}{{\tilde\Phi}}
\newc{\tPsi}{{\tilde\Psi}}
\newc{\tA}{{\tilde A}}
\newc{\talpha}{{\tilde\alpha}}
\newc{\tbeta}{{\tilde\beta}}
\newc{\bA}{{\mathbf A}}
\newc{\bB}{{\bf B}}
\newc{\br}{{\bf r}}
\newc{\sig}{{\mathbf\sigma}}
\newc{\eg}{{\rm e.g.\ }}
\newc{\ie}{{\rm i.e.\ }}
\newcommand{\bey}{\begin{eqnarray}}
\newcommand{\pslash}{\not{\hbox{\kern-2.3pt $p$}}}
\newcommand{\pdslash}{\not{\hbox{\kern-2pt $\partial$}}}
\newcommand{\eey}{\end{eqnarray}}
\newtheorem{theorem}{Theorem}
\newtheorem{proposition}{Proposition}
\newenvironment{proof}[1][Proof]{\noindent\textbf{#1.} }{\ \rule{0.5em}{0.5em}}
\newtheorem{lemma}{Lemma}
\newtheorem{definition}{Definition}

\begin{document}

\begin{titlepage}
\vskip 2cm
\begin{center}
{\Large  Poisson superalgebras and quantization
\footnote{{\tt matrindade@uneb.br}}}
 \vskip 10pt
{ Marco A. S. Trindade  \\}
\vskip 5pt
{\sl Colegiado de Física, Departamento de Ciências Exatas e da Terra, Universidade do Estado da Bahia\\
Rua Silveira Martins, Cabula, 41150000, Salvador, Bahia, Brazil\\}
\vskip 2pt
\end{center}

\begin{abstract}

In this work, we find the Poisson superalgebras related to schemes of quantization. Initially, we consider the Dirac superbracket in the context of the quantization of constrained systems. Next, we show the existence of a Poisson supermanifold in the scenario of quantization deformation.
\end{abstract}

\bigskip


\vskip 3pt

\end{titlepage}


\newpage

\setcounter{footnote}{0} \setcounter{page}{1} \setcounter{section}{0} %
\setcounter{subsection}{0} \setcounter{subsubsection}{0}

\section{Introduction}
The Poisson algebras play a crucial role in Hamilton mechanics \cite{Arnold, Mukunda}. They provide the mathematical structure for describing the temporal evolution of classical physical systems. Furthermore, it also has great relevance in the scope of quantum theory. In quantum mechanics and quantum field theory, the procedure of canonical quantization corresponds to a quantization map in which the Lie algebra of the Poisson brackets is supplanted by the Lie algebra of commutators \cite{Dirac,Weinberg,Kaku}. The basic underlying principle of canonical quantization is that classical and quantum systems are just different realizations of the same algebraic structure, as highlighted by \cite{Loja}. In this spirit, Dirac proposed a generalization of the Poisson bracket allowing canonical quantization for constrained systems \cite{Dirac}.

In deformation quantization \cite{Dito, Cannas}, Poisson algebras also have a fundamental importance. Roughly speaking, an associative algebra of smooth functions is deformed through a star product and we obtain, in first order, a Poisson algebra \cite{Loja}. A geometrical approach for constrained systems is related to Dirac structures which were originally proposed by Courant \cite{Courant, Burs}. The Dirac structure on vector space $V$ is a vector subspace $L\subset V\oplus V^{*}$, which is maximally isotropic under certain geometric pairing. Examples of integrable Dirac structures are Poisson manifolds. \cite{Courant}. One of the most relevant results in the field of quantization deformation is the Kontsevich formula that provided a way for constructing a star product associated with an arbitrary Poisson manifold \cite{Kontsevich}.

A natural extension of the Poisson algebra concept is the Poisson superalgebra \cite{Corwin,Kostant,Shestakov}. For this case, we have a $\mathbb{Z}_{2}$-graded space with two superalgebras: an associative supercommutative superalgebra and a Lie superalgebra with superderivation. Superalgebras are fundamental tools in supersymmetry. An example is the super-Poincaré algebra, which is an extension of Poincaré algebra to take into account supersymmetry \cite{Varad}. In this context, thermal Lie superalgebras have been proposed \cite{Trind}. Another important superalgebra that plays a fundamental role in deformation theory is the Gerstenhaber algebra \cite{Gerst, Loja}. Its difference from the Poisson superalgebra lies in the degree and in the Jacobi identity.

 In this work, we investigated Poisson superalgebras in the context of Dirac quantization of constrained systems and of the quantization deformation with the super-star product. The presentation of this manuscript is the following. In Section 2, we present our results about Poisson superalgebras and Dirac brackets. In Section 3, we show the existence of the Poisson supermanifold through superbracket $B_{1}$. Section 4 is dedicated to our conclusions and perspectives. In Appendix A, we review some concepts related to supermanifolds.

\section{Poisson superalgebra and Dirac superbracket}
In this section, we will show how to construct a Dirac superbracket to get we have a Poisson superalgebra structure. Consider the following definition \cite{Shestakov}
\begin{definition}
A Poisson superalgebra \textit{A} is an associative superalgebra equipped with a Lie superbracket $\{\cdot , \cdot \}$: $A \otimes A \rightarrow A$ such that $(A, \{ \cdot, \cdot \})$ is a Lie superalgebra and a superderivation of A:
\begin{eqnarray}
\{fg,h\}=f\{g,h\}+(-1)^{\partial{(f)}\partial{(g)}}\{f,h\}g,
\end{eqnarray}
where $\partial{(f)}$ denotes the degree of $f$ $(\partial{(f)} \in \mathbb{Z}/2\mathbb{Z})$. 
\end{definition}
In this work we consider that associative superproduct is supercommutative, i. e., 
\begin{eqnarray}
fg=(-1)^{\partial (f) \partial (g)}gf
\end{eqnarray}
 
\begin{definition}
A ideal $A$ is a Poisson superalgebra A is an ideal on the associative superalgebra A such that $\{f,I\}\subseteq I$.
\end{definition}

\begin{proposition}
Every element of an ideal $I$ in a Poisson superalgebra is a linear combination of elements of even degree and odd degree with with non-zero coefficients, i. e., non-homogeneous elements.
\end{proposition}
\begin{proof}
Let $f^{(1)} \in A$ an elements of odd degree and suppose that $I$ contains only elements of odd degree. We have that $\partial (\{f^{(1)},I' \})=\partial (f^{(1)})+ \partial (I')=0 \ (mod \ 2)$, with $I' \in I$.  Therefore $\{f^{(1)},I'\}\not \in I$. On the other hand, consider that $I$ contains only elements of even degree. Then, $\partial (\{f^{(1)},I' \})=\partial (f^{(1)})+ \partial (I')=1 \ (mod \ 2)$ and again,  $\{f^{(1)},I'\}\not \in I$. The same is valid for structure of associative superalgebra and consequently the proposition is proved.
\end{proof}

Notice that if $I$ is an ideal in associative superalgebra A, we have that $I$ is not necessarily ideal in the Poisson superalgebra A. For the Poisson algebras, Hermann \cite{Herm} established conditions under which there is a Poisson algebra $P$ from factor space $P/I_{P}$, where $I_{P}$ is an ideal on the associative algebra. For this, the Dirac bracket was used. Their approach is a purely algebraic description for the Dirac quantization procedure of constrained systems. Our work is in this perspective. 

Let $A= A_{0}+A_{1}$ be an associative superalgebra and let $\phi_{\alpha}$ and $\phi_{\beta}$ be homogeneous parts of generators of ideals of the associative superalgebra $A$. We set $\Phi_{\alpha, \beta}=-\{\phi_{\alpha}, \phi_{\beta}\}$. Let $I^{+(-)}$ be subspaces of the even (+) or odd(-) parts of the ideal $I$.  we will assume that $\Phi_{\alpha, \beta}^{-1}$ exists. We consider the quotient map $A\rightarrow A^{\tau}=A/I^{+(-)}$ defined by
\begin{eqnarray}\label{def1}
\{f^{\tau},g^{\tau}\}^{\tau}=\{f,g\}^{\tau}+ \Phi_{\alpha \beta}^{\tau}\{f, \phi_{\alpha}\}^{\tau} \{\phi_{\beta},g\}^{\tau},
\end{eqnarray}
for homogeneous elements with same degree and
\begin{eqnarray}\label{def2}
\{f^{\tau},g^{\tau}\}^{\tau}=\{f,g\}^{\tau},
\end{eqnarray}
otherwise.
Notice that the eq. (\ref{def1}) defines a Dirac bracket for an associative algebra when  $\phi_{\alpha}$ and $\phi_{\beta}$ are generators of ideals \cite{Dirac, Herm}. 
\begin{lemma}\label{lema1}
The product defined by equations $(\ref{def1})$ and ($\ref{def2})$ satisfies the bilinearity, super skew-symmetry and superderivation.
\end{lemma}
\begin{proof}
The bilinearity is easy to verify. For the super skew-symmetry,
\begin{eqnarray}
\{f^{\tau},g^{\tau}\}^{\tau}&=&\{f,g\}^{\tau}+ \Phi_{\alpha \beta}^{\tau}\{f, \phi_{\alpha}\}^{\tau} \{\phi_{\beta},g\}^{\tau} \nonumber \\
&=&-(-1)^{\partial (f) \partial (g)}\{g,f\}^{\tau}+\Phi_{\alpha \beta}^{\tau}(-1)^{\partial(f) \partial(\phi_{\alpha})}\{\phi_{\alpha},f\}^{\tau}(-1)^{\partial(\phi_{\beta}) \partial (g)}\{g, \phi_{\beta}\}^{\tau} \nonumber \\
&=&-(-1)^{\partial (f) \partial (g)}\{g,f\}^{\tau}+(-1)^{\partial(f)\partial(\phi_{\alpha})+\partial(\phi_{\beta})\partial(g)+\partial{(\{f, \phi_{\alpha}\})}\partial{(\{\phi_{\beta}, g\})}} \nonumber \\
&\times&\Phi_{\alpha \beta}^{\tau}\{g, \phi_{\beta}\}^{\tau}\{\phi_{\alpha},f\}^{\tau}
\end{eqnarray}
We have that
\begin{eqnarray}
\Phi_{\alpha \beta}&=&\left(-\{\phi_{\alpha}, \phi_{\beta}\}\right)^{-1} \nonumber \\
&=&\left[-(-1)(-1)^{\partial(\phi_{\alpha})\partial(\phi_{\beta})}\{\phi_{\beta}, \phi_{\alpha}\}\right]^{-1} \nonumber \\
&=&(-1)\left[(-1)^{\partial(\phi_{\alpha})\partial(\phi_{\beta})}\right]^{-1}\left[ -\{\phi_{\beta}, \phi_{\alpha} \}\right]^{-1} \nonumber \\
&=&-(-1)^{\partial(\phi_{\alpha})\partial(\phi_{\beta})} \Phi_{\beta, \alpha}
\end{eqnarray}
Then
\begin{eqnarray}\label{anti}
\{f^{\tau},g^{\tau}\}^{\tau}&=&-(-1)^{\partial (f) \partial (g)}\{g,f\}^{\tau} \nonumber \\
&-&(-1)^{\partial(\phi_{\alpha})\partial(\phi_{\beta})+\partial(f)\partial(\phi_{\alpha})+\partial(\phi_{\beta})\partial(g)+\partial{(\{f, \phi_{\alpha}\})}\partial{(\{\phi_{\beta}, g\})}} \nonumber \\
&\times & \Phi_{\beta \alpha}^{\tau}\{g, \phi_{\beta}\}^{\tau}\{\phi_{\alpha},f\}^{\tau}
\end{eqnarray}
According to the product defined by equations $(\ref{def1})$ and ($\ref{def2})$, for homogeneous elements of same degree, we have the following cases \\
Case 1) elements of even degree:
\begin{eqnarray}
\{f^{\tau},g^{\tau}\}^{\tau}&=&-\{g,f\}^{\tau}-\Phi_{\beta \alpha}^{\tau}\{g, \phi_{\beta}\}^{\tau}\{\phi_{\alpha},f\}^{\tau} \nonumber \\
&=&-\left(\{g,f\}^{\tau}+\Phi_{\beta \alpha}^{\tau}\{g, \phi_{\beta}\}^{\tau}\{\phi_{\alpha},f\}^{\tau}\right) \nonumber \\
&=&-\{g^{\tau},f^{\tau}\}^{\tau}
\end{eqnarray}
Case 2) elements of odd degree:
\begin{eqnarray}
\{f^{\tau},g^{\tau}\}^{\tau}&=&\{g,f\}^{\tau}+\Phi_{\beta \alpha}^{\tau}\{g, \phi_{\beta}\}^{\tau}\{\phi_{\alpha},f\}^{\tau} \nonumber \\
&=&\{g^{\tau},f^{\tau}\}^{\tau}
\end{eqnarray}
Otherwise,
\begin{eqnarray}
\{f^{\tau},g^{\tau}\}^{\tau}=\{f,g \}^{\tau}=-(-1)^{\partial(f)\partial(g)}\{g^{\tau},f^{\tau}\}^{\tau}
\end{eqnarray}
Therefore, the super skew-symmetry is verified. For the superderivation, we have:
\begin{eqnarray}\label{supder}
\{f^{\tau},(gh)^{\tau}\}^{\tau}&=&\{f,gh\}^{\tau}+ \Phi_{\alpha \beta}^{\tau}\{f, \phi_{\alpha}\}^{\tau} \{\phi_{\beta},gh\}^{\tau} \nonumber \\
&=&\left(\{f,g \}h+(-1)^{\partial{(f)}\partial{(g)}} g \{f, h \} \right)^{\tau} \nonumber \\
&+& \Phi_{\alpha \beta}^{\tau}\{f, \phi_{\alpha}\}^{\tau}\left(\{\phi_{\beta},g\}h+(-1)^{\partial{(\phi_{\beta})}\partial{(g)}}g \{\phi_{\beta},h\}\right)^{\tau} \nonumber \\
&=&\{f,g \}^{\tau}h^{\tau}+\Phi_{\alpha \beta}^{\tau}\{f, \phi_{\alpha}\}^{\tau} 
\{\phi_{\beta}, g \}^{\tau}h^{\tau} \nonumber \\
&+&(-1)^{\partial{(f)} \partial{(g)}}g^{\tau} \{f,h \}^{\tau} \nonumber \\
&+&\Phi_{\alpha \beta}^{\tau}\{f, \phi_{\alpha}\}^{\tau}(-1)^{\partial{(\phi_{\beta})}\partial{(g)}}g^{\tau}\{\phi_{\beta},h\}^{\tau} \nonumber \\
&=&\{f^{\tau}, g^{\tau}\}^{\tau}h^{\tau}+(-1)^{\partial{(f)}\partial{(g)}}g^{\tau}\{f,h\}^{\tau} \nonumber \\
&+&\Phi_{\alpha \beta}^{\tau}\{f,\phi_{\alpha}\}^{\tau}(-1)^{\partial{(\phi_{\beta})}\partial{(g)}}g^{\tau}\{\phi_{\beta},h\} \nonumber \\
&=&\{f^{\tau}, g^{\tau}\}^{\tau}h^{\tau}+(-1)^{\partial{(f)}\partial{(g)}}g^{\tau}\{f,h\}^{\tau} \Phi_{\alpha \beta}^{\tau}\{f,\phi_{\alpha}\}^{\tau} \nonumber \\
&\times &(-1)^{\partial{(\phi_{\beta})}\partial{(g)}+\partial{(\{f, \phi_{\alpha}\})}\partial{(g)}+\partial{(\Phi_{\alpha \beta})}\partial{(g)}}g^{\tau}\{\phi_{\beta},h\}
\end{eqnarray}
Again, we have two cases: \\
Case 1) elements of even degree:
\begin{eqnarray}
\{f^{\tau},(gh)^{\tau}\}^{\tau}&=&\{f^{\tau}, g^{\tau}\}^{\tau}h^{\tau} + g^{\tau}\{f,h \}^{\tau}+\Phi_{\alpha \beta}^{\tau}\{f, \phi_{\alpha}\}^{\tau}g^{\tau}\{\phi_{\beta},h\} \nonumber \\
&=&\{f^{\tau}, g^{\tau}\}^{\tau}h^{\tau}+g^{\tau}\left(\{f,h \}^{\tau}+\Phi_{\alpha \beta}^{\tau}\{f, \phi_{\alpha}\}^{\tau}\{\phi_{\beta},h\} \right) \nonumber \\
&=&\{f^{\tau}, g^{\tau}\}^{\tau}h^{\tau}+g^{\tau}\{f^{\tau}, h^{\tau}\}^{\tau}
\end{eqnarray}
Case 2) elements of odd degree:
\begin{eqnarray}
\{f^{\tau},(gh)^{\tau}\}^{\tau}&=&\{f^{\tau}, g^{\tau}\}^{\tau}h^{\tau} - g^{\tau}\{f,h \}^{\tau}-\Phi_{\alpha \beta}^{\tau}\{f, \phi_{\alpha}\}^{\tau}g^{\tau}\{\phi_{\beta},h\} \nonumber \\
&=&\{f^{\tau}, g^{\tau}\}^{\tau}h^{\tau}-g^{\tau}\left(\{f,h \}^{\tau}+\Phi_{\alpha \beta}^{\tau}\{f, \phi_{\alpha}\}^{\tau}\{\phi_{\beta},h\} \right) \nonumber \\
&=&\{f^{\tau}, g^{\tau}\}^{\tau}h^{\tau}-g^{\tau}\{f^{\tau}, h^{\tau}\}^{\tau}
\end{eqnarray}
Otherwise,
\begin{eqnarray}
\{f^{\tau},(gh)^{\tau}\}^{\tau}&=&\{f,gh \}^{\tau} \nonumber \\
&=&\left(\{f,g \}h+(-1)^{\partial{(f)} \partial{(g)}}g \{f, h \}\right)^{\tau} \nonumber \\
&=&\{f^{\tau},g^{\tau} \}h^{\tau}+(-1)^{\partial{(f)} \partial{(g)}}g^{\tau} \{f^{\tau}, h^{\tau} \}
\end{eqnarray}
Consequently the superderivation is satisfied.
\end{proof}

\begin{theorem}
The associative superalgebra $A$ defined by product (\ref{def1}) and (\ref{def2}) is a Poisson superalgebra.
\end{theorem}
\begin{proof}
By the Lemma \ref{lema1}, we have the super skew-symmetry and superderivation. It remains for us to prove the super Jacobi identity. It is sufficient to consider the case in which product (\ref{def1}) is valid.
\begin{eqnarray}
\{\{f^{\tau}, g^{\tau}\}^{\tau}, h^{\tau}\}^{\tau}&=& \{\{f,g \}^{\tau}+ \Phi_{rs}^{\tau}\{f, \phi_{r}\}^{\tau}\{\phi_{s}, g\}^{\tau},h^{\tau}\}^{\tau} \nonumber \\
&=&\{\{f,g\}+\Phi_{rs}\{f, \phi_{r}\}\{\phi_{s},g\},h\}^{\tau} \nonumber \\
&+&\Phi_{tu}^{\tau}\{\{f,g\}+\Phi_{rs}\{f, \phi_{r}\}\{\phi_{s},g\},\phi_{t}\}^{\tau}\{\phi_{u},h\} \nonumber \\
&=&\{\{f,g\},h\}^{\tau}+(-1)^{\partial(\{\phi_{s},g\})\partial{(h)}}\Phi_{rs}^{\tau}\{\{f,\phi_{r}\},h\}^{\tau}\{\phi_{s},g\}^{\tau} \nonumber \\
&+&(-1)^{\partial{(\{\phi_{s},g\})}\partial{(h)}+\partial{(\{f, \phi_{r}\})}\partial{(h)}}\{\Phi_{r,s},h\}^{\tau}\{f, \phi_{r}\}^{\tau}\{\phi_{s},g\}^{\tau} \nonumber \\
&+&\Phi_{rs}^{\tau}\{f,\phi_{r}\}^{\tau}\{\{\phi_{s},g\},h\}^{\tau}
+ \Phi_{tu}^{\tau}\{\{f,g\},h\}^{\tau}\{\phi_{u},h\}^{\tau} \nonumber \\
&+&\Phi_{tu}^{\tau}(-1)^{\partial{(\{\phi_{s},g\})}\partial{(\phi_{t})}} \Phi_{rs}
\{\{f, \phi_{r}\},\phi_{t}\}^{\tau}\{\phi_{s},g\}^{\tau}\{\phi_{u},h\}^{\tau} \nonumber \\
&+&(-1)^{\partial{(\{\phi_{s},g\})}\partial{(\phi_{t})}+\partial{(\{f,\phi_{r}\})}\partial{(\phi_{t})}}\Phi_{tu}^{\tau}\{\Phi_{rs}, \phi_{t}\}^{\tau}\{f,\phi_{r}\}^{\tau}
\{\phi_{s},g\}^{\tau}\{\phi_{u},h\}^{\tau} \nonumber \\
&+&\Phi_{tu}^{\tau}\Phi_{rs}^{\tau}\{f,\phi_{r}\}^{\tau}\{\{\phi_{s},g\}, \phi_{t}\}^{\tau}
\{\phi_{u},h\}^{\tau}
\end{eqnarray}
The case of even degree coincides with the usual case discussed by Dirac \cite{Dirac} and the case of odd degree can be proved analogously. Let $\sum_{cyc}$ be a sum of three cyclic permutations. For the first term, we have $\Psi_{1}=\sum_{cyc}\{\{f,g\},h\}$, by using the Jacobi identity for elements of odd degree from the Poisson superalgebra. Let $\Psi_{2,4,5}$ be a cyclic sum of second, fourth and fifth terms.
\begin{eqnarray}
\Psi_{2,4,5}&=&\sum_{cyc}(-1)^{\partial(\{\phi_{s},g\})\partial{(h)}}\Phi_{rs}^{\tau}\{\{f,\phi_{r}\},h\}^{\tau}\{\phi_{s},g\}^{\tau} \nonumber \\
&+&\Phi_{rs}^{\tau}\{f,\phi_{r}\}^{\tau}\{\{\phi_{s},g\},h\}^{\tau} \nonumber \\
&+&\Phi_{tu}^{\tau}\{\{f,g\},h\}^{\tau}\{\phi_{u},h\}^{\tau} \nonumber \\
&=&\sum_{cyc}(-1)^{\partial(\{\phi_{s},g\})\partial{(\{\{f,\phi_{r}\},h\})}}\Phi_{rs}^{\tau}\{\phi_{s},g\}^{\tau}\{\{f,\phi_{r}\},h\}^{\tau} \nonumber \\
&+&\Phi_{rs}^{\tau}\{f,\phi_{r}\}^{\tau}\{\{\phi_{s},g\},h\}^{\tau} \nonumber \\
&+&(-1)^{\partial{(\{\{f,g\},h\})\partial{(\{\phi_{u},h\})}}}\Phi_{tu}^{\tau}\{\phi_{u},h\}^{\tau}\{\{f,g\},h\}^{\tau} \nonumber \\
&=&\Phi_{rs}^{\tau}\{\phi_{s},g\}^{\tau}\{\{f,\phi_{r}\},h\}^{\tau}
+\Phi_{rs}^{\tau}\{f,\phi_{r}\}^{\tau}\{\{\phi_{s},g\},h\}^{\tau} \nonumber \\
&+&\Phi_{tu}^{\tau}\{\phi_{u},h\}^{\tau}\{\{f,g\},h\}^{\tau}
\end{eqnarray}
Performing the cyclic permutation of $f,g$ and $h$ in the last two terms from the last equality above and using that $\Phi_{rs}=\Phi_{sr}$ (elements of odd degree), we have
\begin{eqnarray}
\Psi_{2,4,5}&=&\sum_{cyc} \Phi_{rs}^{\tau} \{\phi_{s},g\}^{\tau} \left[\{f, \phi_{r}\},h\}^{\tau}+\{\{\phi_{r},h\},f\}+\{\{h,f\},\phi_{r}\}\right] \nonumber \\
&=&0
\end{eqnarray}
using the Jacobi identity for elements of odd degree. Let $\Psi_{6,8}$ be the cyclic sum of the sixth and eight terms. Then

\begin{eqnarray}
\Psi_{6,8}&=&(-1)^{\partial{(\{\phi_{s},g\})}\partial{(\phi_{t})}}\Phi_{tu}^{\tau} \Phi_{rs}^{\tau}
\{\{f, \phi_{r}\},\phi_{t}\}^{\tau}\{\phi_{s},g\}^{\tau}\{\phi_{u},h\}^{\tau} \nonumber \\
&+&\Phi_{tu}^{\tau}\Phi_{rs}^{\tau}\{f,\phi_{r}\}^{\tau}\{\{\phi_{s},g\}, \phi_{t}\}^{\tau}
\{\phi_{u},h\}^{\tau} \nonumber \\
&=&(-1)^{\partial{(\{\phi_{s},g\})}\partial{(\{f,\phi_{r}\},\phi_{t}\})}+\partial{(\{\phi_{u},h\})}\partial{(\{f,\phi_{r}\},\phi_{t}\})}} \nonumber \\
&\times& \Phi_{rs}^{\tau}\Phi_{tu}^{\tau} 
\{\phi_{s},g\}^{\tau}\{\phi_{u},h\}^{\tau}\{\{f, \phi_{r}\},\phi_{t}\}^{\tau} \nonumber \\
&+&(-1)^{\partial{(\{\{\phi_{s},g\},\phi_{t}\})}\partial{(\phi_{u},h\})}+\partial{(\{f,\phi_{r}\}
)}\partial{(\{\phi_{u},h\})}} \nonumber \\
&\times& \Phi_{rs}^{\tau}\Phi_{tu}^{\tau}\{\phi_{u},h\}^{\tau}\{f,\phi_{r}\}^{\tau}\{\{\phi_{s},g\}, \phi_{t}\}^{\tau} \nonumber \\
&=&\Phi_{rs}^{\tau}\Phi_{tu}^{\tau} 
\{\phi_{s},g\}^{\tau}\{\phi_{u},h\}^{\tau}\{\{f, \phi_{r}\},\phi_{t}\}^{\tau}  \nonumber \\
&+&\Phi_{rs}^{\tau}\Phi_{tu}^{\tau}\{\phi_{u},h\}^{\tau}\{f,\phi_{r}\}^{\tau}\{\{\phi_{s},g\}, \phi_{t}\}^{\tau}
\end{eqnarray}
Considering a cyclic permutation of $r,s,t,u$ and $f,g,h$ in the second term from the last equation, we get
\begin{eqnarray}
\Psi_{68}&=&\sum_{cyc}\Phi_{rs}\Phi_{tu}\{\phi_{s},g\}^{\tau}\{\phi_{u},h\}^{\tau}
\left[\{f, \phi_{r}\}, \phi_{t}\} +\{\{\phi_{t},f\},\phi_{r}\} \right]^{\tau} \nonumber \\
&=&-\sum_{cyc}\Phi_{rs}\Phi_{tu}\{\phi_{s},g\}^{\tau}\{\phi_{u},h\}^{\tau} \{\{\phi_{r},\phi_{t}\},f\}^{\tau}
\end{eqnarray}
using the Jacobi identity for elements of even degree. By using the derivation, we have
\begin{eqnarray}\label{cinq}
\{\Phi_{tu} \{\phi_{r}, \phi_{t}\},f\}&=&-\{\Phi_{tu},f\}\{\phi_{r}, \phi_{t}\}+\Phi_{tu
} \{\{\phi_{r},\phi_{t}\},f\} \nonumber \\
&=&0
\end{eqnarray}
Consequently, since $\Phi_{ss'}^{\tau}\{\phi_{s}, \phi_{s''}\}^{\tau}=\delta_{ss'}$
\begin{eqnarray}
\Psi_{68}&=&-\sum_{cyc}\Phi_{rs}\{\phi_{s},g\}^{\tau}\{\phi_{u},h\}^{\tau}\{\Phi_{tu},f\}^{\tau}\{\phi_{r},\phi_{t}\}^{\tau} \nonumber \\
&=&-\sum_{cyc}\Phi_{rs}\{\phi_{r},\phi_{t}\}^{\tau}\{\phi_{s},g\}^{\tau}\{\phi_{u},h\}^{\tau}\{\Phi_{tu},f\}^{\tau} \nonumber \\
&=&-\sum_{cyc} \delta_{st}\{\phi_{s},g\}^{\tau}\{\phi_{u},h\}^{\tau}\{\Phi_{tu},f\}^{\tau} \nonumber \\
&=&-\sum_{cyc} \{\phi_{t},g\}^{\tau}\{\phi_{u},h\}^{\tau}\{\Phi_{tu},f\}^{\tau} \nonumber \\
&=&-\sum_{cyc}\{\Phi_{rs},h\}^{\tau}\{f,\phi_{r}\}^{\tau}\{\phi_{s},g\}^{\tau},
\end{eqnarray}
performing the cyclic permutation of $f,g$ and $h$. We set 
\begin{eqnarray}
\Psi_{3}=\sum_{cyc}(-1)^{\partial{(\{\phi_{s},g\})}\partial{(h)}+\partial{(\{f, \phi_{r}\})}\partial{(h)}}\{\Phi_{r,s},h\}^{\tau}\{f, \phi_{r}\}^{\tau}\{\phi_{s},g\}^{\tau}
\end{eqnarray}
Therefore $\Psi_{3}+\Psi_{68}=0$. We set $\sum_{rsu}^{r's'u'}$ the sum of three cyclic permutations of $r,s,u,r',s',u'$. Using the Jacobi identity, we have
\begin{eqnarray}\label{ses}
\sum_{rsu}^{r's'u'}\Phi_{r'r}^{\tau}\Phi_{s's}^{\tau}\Phi_{u'u}^{\tau}\{\{\phi_{r'},\phi_{s'}\},\phi_{u'}\}^{\tau}=0
\end{eqnarray}
Using the eq. (\ref{cinq}) with $f$ replaced by $\phi_{u'}$, we get
\begin{eqnarray}
-\{\Phi_{r'r},\phi_{u'}\}^{\tau}\{\phi_{r'},\phi_{s'}\}^{\tau}+\Phi_{r'r}^{\tau}\{\{\phi_{r'},\phi_{s'}\},\phi_{u'}\}^{\tau}=0
\end{eqnarray}
Then, with the help of eq. (\ref{ses}), we have
\begin{eqnarray}
\sum_{rsu}^{r's'u'} \phi_{s's}^{\tau} \phi_{u'u}^{\tau} \{\phi_{r'}, \phi_{s'}\}^{\tau} \{\Phi_{r'r}, \phi_{u'}\}^{\tau}=0
\end{eqnarray}
Using that $\Phi_{ss'}^{\tau}\{\phi_{s}, \phi_{s''}\}^{\tau}=\delta_{ss'}$, we get
\begin{eqnarray}\label{fin}
\sum_{rsu}^{r's'u'} \Phi_{u'u}^{\tau}\{\Phi_{rs},\phi_{u'}\}=0 
\end{eqnarray}
We set
\begin{eqnarray}
\Psi_{7}=\sum_{cyc} (-1)^{\partial{(\{\phi_{s},g\})}\partial{(\phi_{t})}+\partial{(\{f,\phi_{r}\})}\partial{(\phi_{t})}}\Phi_{tu}^{\tau}\{\Phi_{rs}, \phi_{t}\}^{\tau}\{f,\phi_{r}\}^{\tau}
\{\phi_{s},g\}^{\tau}\{\phi_{u},h\}^{\tau}
\end{eqnarray}
By using the eq. (\ref{fin}), we have that $\Psi_{7}=0$ and the super Jacobi identity is proved.
\end{proof}

\begin{lemma}\label{fund}
Let $f,g$ and $h$ be homogeneous elements. Suppose that the superderivation, super skew-symmetry and the super Jacobi identity is satisfied for these homogeneous elements. Then, the super Jacobi identity is valid, i. e.,
\begin{eqnarray}
(-1)^{\partial{(f)}\partial{(h)}}\{f,\{g,h \}\}+(-1)^{\partial{(f)}\partial{(g)}}\{g,\{h,f \}\}+(-1)^{\partial{(g)}\partial{(h)}}\{h,\{f,g \}\}=0
\end{eqnarray}
\end{lemma}

\begin{proof}
It is enough to prove that
\begin{eqnarray}{\label{j1}}
\{b_{1}, \{f_{1}, f_{2}\}\}+\{f_{1}, \{f_{2}, b_{1}\}\}+\{f_{2}, \{b_{1}, f_{1}\}\}=0
\end{eqnarray}
and
\begin{eqnarray}{\label{j2}}
\{b_{1}, \{b_{2}, f_{3}\}\}+\{b_{2}, \{f_{3}, b_{1}\}\}+\{f_{3}, \{b_{1}, b_{2}\}\}=0,
\end{eqnarray}
where $b_{i}$ are elements of even degree and $f_{i}$ are elements of odd degree.
Initially we will prove the equation (\ref{j1}). Notice that every element of even degree can be written as the product of two elements of odd degree, i.e. $b_{1}=\epsilon f_{3}$, where $\partial{(\epsilon)}=\partial{(f_{3})}=1$. Developing each of the terms in the equation (\ref{j1}), we get
\begin{eqnarray}{\label{a1}}
\{ \epsilon f_{3}, \{ f_{1}, f_{2} \}\}=(-1)^{\partial{(f_{3}}) \partial{(\{f_{1},f_{2}\})}}
\{\epsilon, \{f_{1}, f_{2}\}\}f_{3}+\epsilon \{f_{3}, \{f_{1}, f_{2}\}\},
\end{eqnarray}
\begin{eqnarray}{\label{a2}}
\{f_{1}, \{f_{2}, \epsilon f_{3}\}\}&=&(-1)^{\partial{(f_{2})}\partial{(\epsilon})}(-1)^{\partial{(f_{1})}\partial{(\epsilon})} \epsilon \{ f_{1}, \{ f_{2}, f_{3} \} \} \nonumber \\
&+&(-1)^{\partial{(f_{2})}\partial{(\epsilon})}\{f_{1}, \epsilon\} \{f_{2}, f_{3}\} \nonumber \\
&+&(-1)^{\partial{(f_{1})}\partial{(\{f_{2}, \epsilon \})}}\{f_{2}, \epsilon\} \{f_{1}, f_{3}\} \nonumber \\
&+&\{f_{1}, \{f_{2}, \epsilon \}\}f_{3},
\end{eqnarray}
\begin{eqnarray}{\label{a3}}
\{f_{2}, \{\epsilon f_{3}, f_{1}\}\}&=&(-1)^{\partial{(f_{3})}\partial{(f_{1})}}
(-1)^{\partial{(f_{2})}\partial{(\{\epsilon, f_{1}\})}} \{\epsilon, f_{1}\}\{f_{2},f_{3}\} \nonumber \\
&+&(-1)^{\partial{(f_{3})}\partial{(f_{1})}}
(-1)^{\partial{(f_{2})}\partial{(\{\epsilon, f_{1}\})}} \{f_{2}, \{\epsilon, f_{1}\}\}f_{3} \nonumber \\
&+&(-1)^{\partial{(f_{2})} \partial{(\epsilon)}}\epsilon \{f_{2}, \{f_{3},f_{1}\}\} \nonumber \\
&+&\{f_{2}, \epsilon \} \{f_{3}, f_{1}\}
\end{eqnarray}
Adding the equations (\ref{a1}), (\ref{a2}) and (\ref{a3}), we get the equation (\ref{j1}).
For the equation (\ref{j2}), we have
\begin{eqnarray}{\label{l1}}
\{b_{1}, \{b_{2}, f_{3}\}\}&=&\{\epsilon_{1}f_{1}, \{\epsilon_{2}f_{2}, f_{3}\}\} \nonumber \\
&=&(-1)^{\partial{(f_{2})}\partial{(f_{3})}+\partial{(f_{1})} \partial{(\{\epsilon_{2},f_{3}\}f_{2}})+\partial{(\epsilon_{1})} \partial{(\{\epsilon_{2},f_{3}\})}}   \{\epsilon_{2}, f_{3}\}\{\epsilon_{1}, f_{2}\}f_{1} \nonumber \\
&+&(-1)^{\partial{(f_{2})}\partial{(f_{3})}+\partial{(f_{1})}\partial{(\{\epsilon_{2},f_{3}\}f_{2})}} \{\epsilon_{1}, \{\epsilon_{2},f_{3}\}\}f_{2}f_{1} \nonumber \\
&+&(-1)^{\partial{(f_{2})}\partial{(f_{3})}+\partial{(f_{1})}\partial{(\{\epsilon_{2},f_{3}\})}} \epsilon_{1} \{\epsilon_{2}, f_{3}\}\{f_{1}, f_{2}\} \nonumber \\
&+&(-1)^{\partial{(f_{2})}\partial{(f_{3})}} \epsilon_{1}\{f_{1}, \{\epsilon_{2}, f_{3}\}\}f_{2} \nonumber \\
&+&(-1)^{\partial{(f_{1})}\partial{(\epsilon_{2}\{f_{2},f_{3}\})}+\partial{(\epsilon_{1})}\partial{(\epsilon_{2})}}\epsilon_{2}\{\epsilon_{1}, \{f_{2}, f_{3}\}\}f_{1} \nonumber \\
&+&(-1)^{\partial{(f_{1})}\partial{(\epsilon_{2}\{f_{2},f_{3}\})}}\{\epsilon_{1}, \epsilon_{2}\}\{f_{2},f_{3}\}f_{1} \nonumber \\
&+&(-1)^{\partial{(f_{1})}\partial{(\epsilon_{2})}}\epsilon_{1} \epsilon_{2} \{f_{1}, \{f_{2}, f_{3}\}\} \nonumber \\
&+&\epsilon_{1}\{f_{1}, \epsilon_{2}\}\{f_{2}, f_{3}\},
\end{eqnarray}

\begin{eqnarray}{\label{l2}}
\{b_{2}, \{f_{3}, b_{1}\}\}&=&\{\epsilon_{2}f_{2}, \{f_{3}, \epsilon_{1} f_{1}\}\} \nonumber \\
&+&(-1)^{\partial{(f_{2})} \partial{(\epsilon_{1})} + \partial{(f_{2})} \partial{(\epsilon_{1} \{f_{2}, f_{1}\})} + \partial{(\epsilon_{2})} \partial{(\epsilon_{1})}} \epsilon_{1}\{\epsilon_{2}, \{f_{3},f_{1}\}\}f_{2} \nonumber \\
&+&(-1)^{\partial{(f_{2})} \partial{(\epsilon_{1})} + \partial{(f_{2})} \partial{(\epsilon_{1} \{f_{2}, f_{1}\})}} \{\epsilon_{2}, \epsilon_{1}\} \{f_{3},f_{1}\}f_{2} \nonumber \\
&+&(-1)^{\partial{(f_{3})} \partial{(\epsilon_{1})}+\partial{(f_{2})} \partial{(\epsilon_{1})}} \epsilon_{2}\epsilon_{1}\{f_{2},\{f_{3},f_{1}\}\} \nonumber \\
&+&(-1)^{\partial{(f_{3})} \partial{(\epsilon_{1})}}\epsilon_{2}\{f_{2},\epsilon_{1}\}\{f_{3},f_{1}\} \nonumber \\
&+&(-1)^{\partial{(f_{2})} \partial{(\{f_{3}, \epsilon_{1}\} f_{1})}+\partial{(\epsilon_{2})} \partial{(\{f_{3}, \epsilon_{1}\} )}} \{f_{3}, \epsilon_{1}\} \{\epsilon_{2}, f_{1}\}f_{2} \nonumber \\
&+&(-1)^{\partial{(f_{2})} \partial{(\{f_{3}, \epsilon_{1}\} f_{1})}}\{\epsilon_{2},\{f_{3}, \epsilon_{1}\}\}f_{1}f_{2} \nonumber \\
&+&(-1)^{\partial{(f_{2})} \partial{(\{f_{3}, \epsilon_{1}\})}} \epsilon_{2} \{f_{3}, \epsilon_{1}\} \{f_{2}, f_{1}\} \nonumber \\
&+&\epsilon_{2} \{f_{2}, \{f_{3}, \epsilon_{1}\}\} f_{1}
\end{eqnarray}

\begin{eqnarray}{\label{l3}}
\{f_{3},\{b_{1}, b_{2}\}\}&=&\{f_{3}, \{ \epsilon_{1} f_{1}, \epsilon_{2}f_{2} \}\} \nonumber \\
&+&(-1)^{\partial{(f_{1})}\partial{(\epsilon_{2}f_{2})}+ \partial{(\epsilon_{1})}\partial{(\epsilon_{2})}+ \partial{(\epsilon_{2})}+ \partial{(f_{3})}\partial(\{, \epsilon_{1}, f_{2}\})}\epsilon_{2}\{\epsilon_{1},f_{2}\} \{f_{3},f_{1}\} \nonumber \\
&+&(-1)^{\partial{(f_{1})}\partial{(\epsilon_{2}f_{2})}+ \partial{(\epsilon_{1})}\partial{(\epsilon_{2})}+ \partial{(\epsilon_{2})}}\{f_{2}\epsilon_{2}\}\{\epsilon_{1},f_{2}\}f_{1} \nonumber \\
&+&(-1)^{\partial{(f_{1})}\partial{(\epsilon_{2}f_{2})}+ \partial{(\epsilon_{1})}\partial{(\epsilon_{2})}}\{f_{3}, \epsilon_{2}\}\{\epsilon_{1}, f_{2}\}f_{1} \nonumber \\
&+&(-1)^{\partial{(f_{1})}\partial{(\epsilon_{2}f_{2})}+ \partial{(f_{3})}\partial{(\{ \epsilon_{1}, \epsilon_{2}\}f_{2})}}\{\epsilon_{1}, \epsilon_{2}\}f_{2}\{f_{3},f_{1}\} \nonumber \\
&+&(-1)^{\partial{(f_{1})}\partial{(\epsilon_{2}f_{2})}+ \partial{(f_{3})}\partial{(\{ \epsilon_{1}, \epsilon_{2}\})}}\{\epsilon_{1}, \epsilon_{2}\}\{f_{3},f_{2}\}f_{1} \nonumber \\
&+&(-1)^{\partial{(f_{1})}\partial{(\epsilon_{2}f_{2})}}\{f_{3},\{\epsilon_{1}, \epsilon_{2}\}\}f_{2}f_{1} \nonumber \\
&+&(-1)^{\partial{(f_{1})}\partial{(\epsilon_{2}f_{2})} + \partial{(f_{1})}\partial{(\epsilon_{2})}+\partial{(f_{3})}\partial{(\epsilon_{1}\epsilon_{2})}}
\epsilon_{1}\epsilon_{2}\{f_{3}, \{f_{1}, f_{2}\}\} \nonumber \\
&+&(-1)^{\partial{(f_{1})}\partial{(\epsilon_{2}f_{2})} + \partial{(f_{1})}\partial{(\epsilon_{2})}+\partial{(f_{3})}\partial{(\epsilon_{1})}}
\epsilon_{1}\{f_{3}, \epsilon_{2}\}\{f_{1},f_{2}\} \nonumber \\
&+&(-1)^{\partial{(f_{1})}\partial{(\epsilon_{2}f_{2})} + \partial{(f_{1})}\partial{(\epsilon_{2})}+\partial{(f_{3})}\partial{(\epsilon_{1})}}
\{f_{3}, \epsilon_{1}\} \epsilon_{2}\{f_{1}, f_{2}\} \nonumber \\
&+&(-1)^{\partial{(f_{1})}\partial{(\epsilon_{2}f_{2})} + \partial{(f_{3})}\partial{(\epsilon_{1}\{f_{1},\epsilon_{2}\})} }
\epsilon_{1}\{f_{1}, \epsilon_{2}\}\{f_{3}, f_{2}\} \nonumber \\
&+&(-1)^{\partial{(f_{1})}\partial{(\epsilon_{2}f_{2})} + \partial{(f_{3})}\partial{(\epsilon_{1})} }
\epsilon_{1}\{f_{3}, \{f_{1}, \epsilon_{2}\}\}f_{2} \nonumber \\
&+&(-1)^{\partial{(f_{1})}\partial{(\epsilon_{2}f_{2})} }
\{f_{3}, \epsilon_{1}\}\{f_{1}, \epsilon_{2}\}f_{2}
\end{eqnarray}
Adding the equations (\ref{l1}), (\ref{l2}) and (\ref{l3}), we get the equation (\ref{j2}). Therefore the super Jacobi identity is proved.
\end{proof}

A natural question is under what conditions the product (\ref{def1}) defines a Poisson superalgebra without the condition (\ref{def2}). The next theorem establishes more general conditions.

\begin{theorem}
We set $\Xi_{\alpha \beta} \equiv \partial{(f)}\partial{(\phi_{\alpha})}+\partial{(\phi_{\beta})}\partial{(g)}+\partial{(f)}\partial{(\phi_{\beta})}+\partial{(\phi_{\alpha})}\partial{(g)}$. The product (\ref{def1}) defines a Poisson superalgebra for homogeneous elements $\phi_{\alpha} (\phi_{\beta})$ since $\Xi_{\alpha, \beta}=2n$, for $n \in \mathbb{N}$.
\end{theorem}. 
\begin{proof}
Inially, we will analyze the superderivation. So we consider the eq. (\ref{supder}). We set
\begin{eqnarray}
\Delta_{\alpha \beta} \equiv \partial{(\phi_{\beta})}\partial{(g)}+\partial{(\{f, \phi_{\alpha }\})}\partial{(g)}+\partial{(\Phi_{\alpha \beta})}\partial{(g)}
\end{eqnarray}
Then
\begin{eqnarray}
\Delta_{\alpha \beta}&=& \partial{(\phi_{\beta})}\partial{(g)}+[\partial{(f)}+\partial{(\phi_{\alpha})}]\partial{(g)}+[\partial{(\phi_{\alpha})}+\partial{(\phi_{\beta})}]\partial{(g)} \nonumber \\
&=&2[\partial{(\phi_{\beta})}\partial{(g)}+\partial{(\phi_{\alpha})}\partial{(g)}]+\partial{(f)}\partial{(g)} 
\end{eqnarray}
so that
\begin{eqnarray}
(-1)^{\Delta_{\alpha \beta}}=(-1)^{\partial{(\phi_{\beta})}\partial{(g)}+\partial{(\{f, \phi_{\alpha }\})}\partial{(g)}+\partial{(\Phi_{\alpha \beta})}\partial{(g)}}=(-1)^{\partial{(f)}\partial{(g)} }
\end{eqnarray}
Consequently
\begin{eqnarray}
\{f^{\tau}, (gh)^{\tau}\}^{\tau}&=&\{f^{\tau},g^{\tau}\}^{\tau}h^{\tau}+(-1)^{\partial{(f)}\partial{(g)}}g^{\tau}\left[\{f,h\}^{\tau}+\Phi_{\alpha \beta}^{\tau} \{f, \phi_{\alpha}\}^{\tau}\{\phi_{\beta},h\}^{\tau}\right] \nonumber \\
&=&\{f^{\tau},g^{\tau}\}^{\tau}h^{\tau}+(-1)^{\partial{(f)}\partial{(g)}}g^{\tau}\{f^{\tau},h^{\tau}\}^{\tau}
\end{eqnarray}
and the superiderivation is satisfied. For the super skew-symmetry, considering the eq. (\ref{anti}), we define
\begin{eqnarray}
\Delta^{*}_{\alpha \beta}\equiv \partial{(f)}\partial{(\phi_{\alpha})}+\partial{(\phi_{\beta})}\partial{(g)}+\partial(\{f, \phi_{\alpha}\})\partial{(\{\phi_{\beta},g\})}+\partial{(\phi_{\alpha})}\partial{(\phi_{\beta})}
\end{eqnarray}
Then
\begin{eqnarray}
\Delta^{*}_{\alpha \beta}&=&\partial{(f)}\partial{(g)}+[\partial{(f)}+\partial{(\phi_{\alpha})}][\partial{(\phi_{\beta})}+ \partial{(g)}]+ \partial(\phi_{\alpha})\partial(\phi_{\beta}) \nonumber \\
&=&\partial{(f)}\partial{(g)}+2\partial{(\phi_{\alpha})}\partial{(\phi_{\beta})}+ \Xi_{\alpha \beta}
\end{eqnarray}
which implies that
\begin{eqnarray}
(-1)^{\Delta^{*}_{\alpha \beta}}=(-1)^{\partial{(f)}\partial{(g)}+ \Xi_{\alpha \beta}}
\end{eqnarray}
Therefore, 
\begin{eqnarray}
\{f^{\tau},g^{\tau}\}^{\tau}&=&-(-1)^{\partial{(f)} \partial{(g)}}\{g,f\}^{\tau} \nonumber \\
&-&(-1)^{\partial(\phi_{\alpha})\partial(\phi_{\beta})+\partial(f)\partial(\phi_{\alpha})+\partial(\phi_{\beta})\partial(g)+\partial{(\{f, \phi_{\alpha}\})}\partial{(\{\phi_{\beta}, g\})}} \nonumber \\
&\times & \Phi_{\beta \alpha}^{\tau}\{g, \phi_{\beta}\}^{\tau}\{\phi_{\alpha},f\}^{\tau} \nonumber \\
&=&-(-1)^{\partial (f) \partial (g)}\{g,f\}^{\tau}-(-1)^{\Delta^{*}_{\alpha \beta}}\Phi_{\beta \alpha}^{\tau}\{g, \phi_{\beta}\}^{\tau}\{\phi_{\alpha},f\}^{\tau} \nonumber \\
&=&-(-1)^{\partial (f) \partial (g)}\{g,f\}^{\tau}-(-1)^{\partial{(f)}\partial{(g)}+ \Xi}\Phi_{\beta \alpha}^{\tau}\{g, \phi_{\beta}\}^{\tau}\{\phi_{\alpha},f\}^{\tau} \nonumber 
\end{eqnarray}
If $\Xi_{\alpha \beta}=2n$, with $n \in \mathbb{N}$, we have 
\begin{eqnarray}
\{f^{\tau},g^{\tau}\}^{\tau}&=&-(-1)^{\partial (f) \partial (g)}\left[\{g,f\}^{\tau} +\Phi_{\beta \alpha}^{\tau}\{g, \phi_{\beta}\}^{\tau}\{\phi_{\alpha},f\}^{\tau}\right] \nonumber \\
&=&-(-1)^{\partial (f) \partial (g)}\{g^{\tau},f^{\tau}\}^{\tau}
\end{eqnarray}
Using the lemma \ref{fund}, the theorem is proved.
\end{proof}

\section{Poisson superalgebra and quantization deformation}
In this section we will address the Poisson superalgebra in the context of deformation quantization \cite{Loja}. More specifically, we will show the existence of a Poisson supermanifold related to star product.

\begin{definition}
Let $M$ be a smooth supermanifold and let A be an associative superalgebra of smooth superfunctions on $M$. A superstar superproduct on $A$ is a bilinear operation: $\star_{\hbar}: A[[\hbar]]\times A[[\hbar]]\rightarrow A[[\hbar]$ such that \\

(i) $\star_{\hbar}$ is $\mathbb{R}[[\hbar]]$-linear: 
\begin{eqnarray}
(\sum_{k \geq 0}f_{k}\hbar^{k}) \star_{\hbar}(\sum_{l \geq 0}g_{l}\hbar^{k})&=& \sum_{k,l \geq 0}(f_{k} \star_{\hbar}g_{l}) \hbar^{k+l};
\end{eqnarray} 
\\
(ii)  $\star_{\hbar}$ is associative:
\begin{eqnarray}
(f \star_{\hbar} g) \star_{\hbar}h= f \star_{\hbar} (g \star_{\hbar} h);
\end{eqnarray}
\\
(iii)  $\star_{\hbar}$ deforms the usual product
\begin{eqnarray}
f \star_{\hbar} g= fg+O(\hbar)&=&(-1)^{\partial (f) \partial (g)} g \star_{\hbar} f \nonumber \\
&=&(-1)^{\partial (f) \partial (g)} gf + O(\hbar); 
\end{eqnarray} 
\\
(iv) $\star_{\hbar}$ is local: 
\begin{eqnarray}
f \star_{\hbar} g=\sum_{k \geq 0} D_{k}(f,g) \hbar^{k},
\end{eqnarray}
where $D_{k}(f,g)$ are bi-differential operators.

\end{definition}

\begin{lemma}\label{Loja}{(paraphrased from Lemma 1.7 of \cite{Loja})} The bi-linear map $D_{1}:A \times A \rightarrow A$ satisfies:
\begin{eqnarray}
fD_{1}(g,h)-D_1(fg,h)+D_{1}(f,gh)-D_{1}(f,g)h=0 \label{L1}
\end{eqnarray}
\end{lemma}

\begin{theorem}
A smooth supermanifold M is a Poisson manifold through the superbracket $D_{1}(f,g)$ defined by the star-superproduct $\star_{\hbar}$.
\end{theorem}
\begin{proof}
The super skew-symmetry follows from the definition. Using the previous Lemma and analogously to the proof of Lemma 1.10 of \cite{Loja}, we obtain:
\begin{eqnarray}
fD_1(g,h)-D_{1}(fg,h)+D_{1}(f,gh)-D_{1}(f,g)h&=&0 \label{qd1} \\
gD_{1}(h,f)-D_{1}(gh,f)+D_{1}(g,hf)-D_{1}(g,h)f&=&0 \label{qd2} \\
hD_{1}(f,g)-D_{1}(hf,g)+D_{1}(h,fg)-D_{1}(h,f)g&=&0 \label{qd3}
\end{eqnarray}
permuting $f, g$ and $h$ in the eq. (\ref{L1}). For the super-derivation, we consider the four cases:

(i) If $f,g$ and $h$ are elements of even degree, doing (\ref{qd2})+(\ref{qd3})-(\ref{qd1}), we get 
\begin{eqnarray}
2fD_{1}(g,h)+2D_{1}(h,fg)-2gD_{1}(h,f)=0  
\end{eqnarray}
\begin{eqnarray}
D_{1}(fg,h)=fD_{1}(g,h)+D_{1}(f,h)g  
\end{eqnarray}

(ii) If $f,g$ and $h$ are elements of odd degree, doing (\ref{qd2})+(\ref{qd3})-(\ref{qd1}), we have 
\begin{eqnarray}
2fD_{1}(g,h)+2D_{1}(h,fg)-2D_{1}(h,f)g=0  
\end{eqnarray}
\begin{eqnarray}
D_{1}(h,fg)=-fD_{1}(h,g)+D_{1}(h,f)g  
\end{eqnarray}

(iii) If $f,g$ are elements of even degree and $h$ is element of odd degree, doing (\ref{qd2})+(\ref{qd3})-(\ref{qd1}), we get 
\begin{eqnarray}
-2D_{1}(fg,h)+2fD_{1}(g,h)-2gD_{1}(h,f)=0
\end{eqnarray}
\begin{eqnarray}
D_{1}(fg,h)=D_{1}(f,h)g+fD_{1}(g,h)=0
\end{eqnarray}

(iv) If $f,g$ are elements of odd degree and $h$ is an element of even degree, doing (\ref{qd1})+(\ref{qd2})+(\ref{qd3}) we have
\begin{eqnarray}
-2D_{1}(fg,h)+2fD_{1}(g,h)-2D_{1}(h,f)g=0
\end{eqnarray}
\begin{eqnarray}
D_{1}(fg,h)=D_{1}(f,h)g+fD_{1}(g,h)
\end{eqnarray}

Summarizing these results:
\begin{eqnarray}
D_{1}(f,gh)=D_{1}(f,g)h+(-1)^{\partial (f) \partial (g)}gD_{1}(f,h)
\end{eqnarray}

For the super Jacobi identity, we consider the supercommutator:
\begin{eqnarray}
[f,g]=f \star g -(-1)^{\partial (f) \partial (g)} g \star f.
\end{eqnarray}
Then
\begin{eqnarray}
[[f,g],h]&=&[f \star g) -(-1)^{\partial (f) \partial (g)}(g \star f)] \star h \nonumber \\
&-&(-1)^{\partial[f \star g-(-1)^{\partial (f) \partial (g)}g \star f]\partial (h)}h \star [ f \star g-(-1)^{\partial (f) \partial (g)}g \star f] \nonumber \\
&=&(f \star g) \star h -(-1)^{\partial (f) \partial (g)} (g \star f) \star h \nonumber \\
&-&(-1)^{\partial[f \star g-(-1)^{\partial (f) \partial (g)}g \star f]\partial (h)} h \star(g \star f) \nonumber \\
&+& (-1)^{\partial[f \star g-(-1)^{\partial (f) \partial (g)}g \star f]\partial (h)}(-1)^{\partial (f) \partial (g)} h \star(g \star f)
\end{eqnarray}
Suppose that $f, g$ and $h$ are elements of odd degree. Consequently
\begin{eqnarray}
[[f,g],h]=(f \star g) \star h-(g \star f) \star h - h \star(f \star g)- h \star (g \star f),
\end{eqnarray}
which implies
\begin{eqnarray}
[[f,g], h]&=&[D_{1}(D_{1}(g, f),h)-D_{1}(D_{1}(g, f),h)-D_{1}(h, D_{1}(g,f) \nonumber \\
&-&D_{1}(h, D_{1}(g,f)]\hbar^{2}+O(\hbar^{3}) \nonumber \\
&=&-2D_{1}(h, D_{1}(g,f)\hbar^{2}+O(\hbar^{3})
\end{eqnarray}
Analogously, Suppose that $f, g$ and $h$ are elements of odd degree. In this case, we have
\begin{eqnarray}
[[f,g], h]&=&(f \star g) \star h-(g \star f) \star h - h \star(f \star g)+ h \star (g \star f) \nonumber \\
&=&4D_{1}(h, D_{1}(g,f)\hbar^{2}+O(\hbar^{3})
\end{eqnarray}
By using the Lemma (\ref{Loja}) and associativity of star product, we get a super Jacobi identity.
\begin{eqnarray}
(-1)^{\partial(f) \partial(g)}D_{1}(D_{1}(f, g),h)+(-1)^{\partial(f) \partial(h)}D_{1}(D_{1}(g, h),f)+(-1)^{\partial(g) \partial(h)}D_{1}(D_{1}(h, f),g)=0. \nonumber
\end{eqnarray}
\end{proof}

\section{Conclusions}
 The scheme of Dirac quantization for the constrained systems is related to Dirac bracket of quotient spaces $\textit{A/I}$, where the constraints define ideals in an associative algebra $\textit{A}$. We consider an associative superalgebra and we redefine the Dirac product so that quotient space has a Poisson superalgebra structure. The development of this procedure led to an interesting lemma which states that if a product satisfies derivation, skew-symmetry and Jacobi identity for elements of degree even or odd, then the super Jacob identity is satisfied. It is worth noting the super Jacobi identity could be also proved analogously to Dirac's original proof of Jacobi identity for the Dirac brackets. In another perspective, we investigated the structure of Poisson superalgebras in the context of quantization deformation. Particularly, we show that the $D_{1}$ product related to the superstar product guarantees the existence of a Poisson supermanifold. As perspectives, we intend to extend the results to Hopf algebras and explore applications in quantum field theory.

\section{Appendix A}
Roughly speaking, one defines a supermanifold as a manifold in which we have odd and even coordinate functions that are smooth. In this appendix, we review some concepts related to this topic. We will use the references \cite{Ten, Can}

\begin{definition} A presheaf $F$ of sets on a topological space $X$ consists of the two data: \\
a) for each open $U$ of $X$, a set $F(U)$, \\
b) for a each pair of open sets $V \subseteq U$ of $X$, a restriction map $\varepsilon_{V}^{U}:  F(U)\rightarrow F(V)$ such that 
\begin{enumerate}
   \item for all $U$, $\varepsilon_{U}^{U}=id_{U}$ 
   \item whenever $W \subseteq V \subseteq U$, $\varepsilon_{W}^{U}=\varepsilon_{W}^{V}\circ \varepsilon_{V}^{U}$, ($U,V,W$ are open sets)  
   \end{enumerate}
\end{definition}

\begin{definition}
A directed set $\Lambda$ is a set with pre-order $\leq$ (that is, a reflexive and transitive relation $\alpha \leq \alpha$, and $\alpha \leq \beta \leq \gamma \Rightarrow \alpha \leq \gamma$), which also satisfies (i) $\forall \alpha, \beta \in \Lambda, \exists \gamma \in \Lambda$ such that $\alpha \leq \gamma$ and $\beta \leq \gamma$.
\end{definition}

We set $\Lambda_{1}=\{(\alpha, \beta) \in \Lambda \times \Lambda; \alpha \leq \beta\}$

\begin{definition}
A direct system of sets indexed by a direct $\Lambda$ is a family of $(U_{\alpha})_{\alpha \in \Lambda}$ of sets with, for each $(\alpha, \beta) \in \Lambda_{1}$, a map of sets $\epsilon_{\alpha \beta}: U_{\alpha}\rightarrow U_{\beta}$ satisfying
\begin{enumerate}
\item $\forall \alpha \in \Lambda$, $\epsilon_{\alpha \alpha}=id_{U_{\alpha}}$
\item $\forall \alpha, \beta, \gamma \in \Lambda$ if $\alpha \leq \beta \leq \gamma$, $\epsilon_{\alpha \gamma}=\epsilon_{\beta \gamma}\circ \epsilon_{\alpha \beta}$
\end{enumerate}
\end{definition}

\begin{definition}
A presheaf of sets over $X$ is a sheaf of sets if the following conditions are satisfied:
\begin{enumerate}
\item Suppose that $U$ is an open of $X$ and $U=\bigcup_{\lambda \in \Lambda} U_{\lambda}$ is an open covering of $U$, and $s,s' \in F(u)$ are two sections of $\epsilon$ such that $\forall \lambda \in \Lambda$, $\epsilon_{U_{\lambda}}^{U}(s)=\epsilon_{U_{\lambda}}^{U}(s')$, then $s=s'$.
\item Suppose that $U$ is open in $X$ and $U=\bigcup_{\lambda \in \Lambda}U_{\lambda}$ is an open covering of $U$; suppose we are given a family $(s_{\lambda})_{\lambda \in \Lambda}$ of sections of $F$ with $\forall \lambda \in \Lambda$, $s_{\lambda} \in F(U_{\lambda})$, such that $\forall \lambda, \mu \in \Lambda$, 
\begin{eqnarray}
\varepsilon_{U_{\lambda}\bigcap U_{\mu}}^{U_{\lambda}}(s_{\lambda})=\varepsilon_{U_{\lambda}\bigcap U_{\mu}}^{U_{\mu}}(s_{\mu})
\end{eqnarray}
then there is $s \in F(u)$ such that $\lambda \in \Lambda$, $\varepsilon_{U_{\lambda}}^{U}(s)=s_{\lambda}$.
\end{enumerate}
\end{definition}
\begin{definition}
Let $M$ be a differentiable manifold an let $A$ be a sheaf of associative supercommutative superalgebras on $M$. A pair $(M, A)$ is a supermanifold if the following conditions are satisfied:
\begin{enumerate}
\item there is a surjective mapping of sheaves $H:A\rightarrow C^{\infty}$;
\item $M$ admits an open cover $(U_{\alpha})$ and there is a real linear space $V$ such that for each $\alpha$ there is a superalgebra isomorphism $I_{\alpha}: A(U_{\alpha}) \rightarrow C_{\infty}(U_{\alpha}) \otimes \Lambda (V)$, where $\Lambda(V)$ is the exterior algebra of space $V$.
\end{enumerate}
\end{definition}
We denote by $(m,n)$ the dimension of supermanifold if $dim M=m$ and $dim V=n$. 
\begin{definition}
A Poisson supermanifold is a manifold $M$ equipped with a sheaf of Poisson superalgebras $(A, \{\cdot, \cdot \})$, such that $(M,A)$ is a supermanifold.
\end{definition}

\section{Acknowledgments}
We thank S. Floquet for stimulating discussions.



\begin{thebibliography}{99}

\bibitem{Arnold} V. I. Arnold, Mathematical Methods of Classical Mechanics, Graduate texts in Mathematics, Springer Verlag, New York, 1974.

\bibitem{Mukunda} E. C. G. Sudarshan and N. Mukunda, Classical Dynamics; A Modern Perspective, John Wiley and Sons, New York, 1974.

\bibitem{Dirac} P. A. Dirac, Generalized Hamiltonian dynamics, Canad. J. Math., 2, 129, 1950.

\bibitem{Weinberg} S. Weinberg, The Quantum Theory of Fields, vol. 1, Cambridge University Press, Cambridge, 1995.

\bibitem{Kaku} M. Kaku, Quantum Field Theory. A modern Introduction, Oxford University Press, New York, 1993.

\bibitem{Loja} R. L. Fernandes, Deformation Quantization and Poisson Geometry, Resenhas IME-USP, Vol.4, 3, 327-361, 2000.

\bibitem{Dito} G. Dito, D. Sternheimer, Deformation Quantization: Genesis, Developments and Metamorphoses, Proceedings of the meeting between mathematicians and theoretical physicists, Strasbourg, 2001. IRMA Lectures in Math. Theoret. Phys., vol. 1, 9-54 Walter De Gruyter, Berlin, 2002.

\bibitem{Cannas} A. C. da Silva, K. Hartshorn and A. Weinstein, Lecture on geometric models of noncommutative algebras, University of California at Berckeley, 1998.

\bibitem{Courant} T. Courant, Dirac manifolds, Trans. Amer. Math. Soc., 319, 631-661, 1990.

\bibitem{Burs} H. Bursztyn, M. Crainic, A. Weinstein, X. Zhu, Integration of twisted Dirac brackets, Duke. Math. J., 123, 564-607, 2004.

\bibitem{Kontsevich} M. Kontsevich, Deformation Quantization of Poisson manifolds, Lett. Math. Phys., 66, 157-216, 2003.

\bibitem{Corwin} L. Corwin, Y. Ne'eman and S. Sternberg, Rev. Mod. Phys., 47, 573, 1975.

\bibitem{Kostant} B. Konstant and S. Sternberg, Symplectic reduction, BRS cohomology and infinite-dimensional Clifford algebras, Ann. Phys., 176, 49-113, 1987.

\bibitem{Shestakov} I. P. Shestakov, Quantization of Poisson Superalgebras and speciality of Jordan Poisson superalgebras, Algebra and Logic, Vol 32, 5, 1993.

\bibitem{Varad} V. S. Varadarajan, Supersymmetry for Mathematicians: An Introduction. American Mathematical Society, Providence, 2004.

\bibitem{Trind} M. A. S. Trindade, Thermal Lie Superalgebras, Int. J. Mod. Phys. A, 36, 30, 2150237, 2021.

\bibitem{Gerst} M. Gerstenhaber, The Cohomology structure of a associative ring, Ann.  Math., 78, 2, 267-288, 1963.

\bibitem{Herm} R. Hermann, Lie Algebras and Quantum Mechanics, Mathematics Lecture Note Series, W. A. Benjamin, INC, New York, 1970.

\bibitem{Ten} B. R. Tennison, N. J. Hitchin, Sheaf Theory, London Mathematical Society, Lecture Note Series, Cambridge University Press, Cambridge, 1976.

\bibitem{Can} F. Canttrijn and L. A. Ibort, Introduction to Poisson supermanifolds, Differential Geometry and its applications, North-Holland, 1, 133-152, 1991.


\end{thebibliography}
\end{document}